\documentclass[a4paper,UKenglish,autoref]{lipics-v2019}

\usepackage[ruled,vlined,linesnumbered,commentsnumbered]{algorithm2e}
\usepackage{xcolor}
\usepackage{hyperref}

\title{A 1.5-approximation algorithm for activating \ \ \ \ \ \ \ 2 disjoint {\em st}-paths}

\titlerunning{A 1.5-pproximation algorithms for activating 2 disjoint {\em st}-paths}

\author{Zeev Nutov}{The Open University of Israel}{nutov@openu.ac.il}
{https://orcid.org/0000-0002-6629-3243}{}
\author{Dawod Kahba}{The Open University of Israel}{kabhad82@gmail.com}{}{}

\authorrunning{Zeev Nutov and Dawod Kahba}

\Copyright{Zeev Nutov and Dawod Kahba}

\nolinenumbers


\acknowledgements{}

\begin{document}

\maketitle
\newcommand {\ignore} [1] {}

\newcommand{\sem}    {\setminus}
\newcommand{\subs}   {\subseteq}
\newcommand{\empt}   {\emptyset}

\newcommand{\f}   {\frac}

\def\al   {\alpha}
\def\be  {\beta}

\def\FF {{\cal F}}
\def\II    {{\cal I}}

\def\atdpa   {{\sc Activation $2$-DP Augmentation}}

\keywords{disjoint $st$-paths, activation problem, minimum power}

\begin{abstract}
In the {\sc Activation $k$ Disjoint $st$-Paths} ({\sc Activation $k$-DP}) problem  we are given a graph $G=(V,E)$ 
with activation costs $\{c_{uv}^u,c_{uv}^v\}$ for every edge $uv \in E$, a source-sink pair $s,t \in V$, and an integer $k$. 
The goal is to compute an edge set $F \subs E$  of $k$ internally node disjoint $st$-paths of minimum activation cost
$\displaystyle \sum_{v \in V}\max_{uv \in E}c_{uv}^v$.
The problem admits an easy $2$-approximation algorithm. 
Alqahtani \& Erlebach \cite{AE} claimed that {\sc Activation $2$-DP} admits a $1.5$-approximation algorithm.
The proof of \cite{AE} has an error, and we will show that the approximation ratio of their algorithm is at least $2$.
We will then give a different algorithm with approximation ratio $1.5$. 
\end{abstract}

\section{Introduction} \label{s:intro}

In network design problems one seeks a cheap subgraph that satisfies a prescribed property.
A traditional setting is when each edge or node has a cost, and we want to minimize the cost of the subgraph.  
This setting does not capture many wireless networks scenarios, where a communication between two nodes 
depends on our ''investment'' in these nodes -- like transmission energy and different types of equipment,
and the cost incurred is a sum of these ``investments''.
This motivates the type of problems we study here.

More formally, in {\bf activation network design problems} we are given 
an undirected (multi-)graph $G=(V,E)$ where every edge $e=uv \in E$ has two (non-negative) {\bf activation costs} $\{c_e^u,c_e^v\}$;
here $e=uv \in E$ means that the edge $e$ has ends $u,v$ and belongs to $E$. 
An edge $e=uv \in E$ is {\bf activated by a level assignment} $\{l_v: v \in V\}$ to the nodes if $l_u \geq c_e^u$ and $l_v \geq c_e^v$. 
The goal is to find a level assignment of minimum value $l(V)=\sum_{v \in V} l_v$, such that 
the activated edge set $F=\{e=uv \in E:c_e^u \leq l_u, c_e^v \leq l_v\}$ satisfies a prescribed property.
Equivalently, the minimum value level assignment that activates an edge set $F \subs E$ is given by $\ell_F(v)=\max \{c_e^v: e \in \delta_F(v)\}$;
here $\delta_F(v)$ denotes the set of edges in $F$ incident to $v$, and a maximum taken over an empty set is assumed to be zero.
We seek an edge set $F \subs E$ that satisfies the given property and minimizes $\ell_F(V)=\sum_{v \in V} \ell_F(v)$.
Note that while we use $l_v$ to denote a level assignment to a node $v$, 
we use a slightly different notation $\ell_F(v)$ for the function that evaluates the optimal 
assignment that activates a given edge set $F$. 

Two types of activation costs were extensively studied in the literature, see a survey \cite{N-sur}.
\begin{itemize}
\item
{\bf Node weights}.  For all $v \in V$, $c_e^v$ are identical for all edges incident to $v$. 
This is equivalent to having node weights $w_v$ for all $v \in V$. The goal is to find a node subset
$V' \subs V$ of minimum total weight $w(V') = \sum_{v \in V'}w_v$ such that the subgraph induced by
$V'$ satisfies the given property.
\item
{\bf Power costs}: For all $e = uv \in E$, $c_e^u=c_e^v$.
This is equivalent to having ``power costs'' $c_e = c_e^u=c_e^v$ for all $e = uv \in E$. 
The goal is to find an edge subset $F \subs E$ of minimum total power 
$\sum_{v \in V} \max\{c_e:e \in \delta_F(v)\}$ that satisfies the given property.
\end{itemize}

Node weighted problems include many fundamental problems such as 
{\sc Set Cover}, {\sc Node-Weighted Steiner Tree}, and {\sc Connected Dominating Set} c.f. \cite{S,KR,GH}.
Min-power problems were studied already in the 90's, c.f. \cite{SRS,WNE,RM,KKKP}, followed by many more.
They were also widely studied in directed graphs, 
usually under the assumption that to activate an edge one needs to assign power only to its tail, 
while heads are assigned power zero, c.f. \cite{KKKP,SM,N-powcov,HKMN,N-sur}. 
The undirected case has an additional requirement - we want the network to be bidirected,
to allow a bidirectional communication.
The general activation setting was first suggested by Panigrahi \cite{P} in 2011. 
Here we use a simpler but less general setting suggested in \cite{KNS}, 
which is equivalent to that of Panigrahi \cite{P} for problems 
in which inclusion minimal feasible solutions have no parallel edges.

In the traditional edge-costs scenario, 
a fundamental problem in network design is the {\sc Shortest $st$-Path} problem.
A natural generalization and the simplest high connectivity network design problem 
is finding a set of $k$ disjoint $st$-paths of minimum edge cost. 
Here the paths may be edge disjoint -- the {\sc $k$ Edge Disjoint $st$-Paths} problem, 
or internally (node) disjoint -- the {\sc $k$ Disjoint $st$-Paths}  problem.
Both problems can be reduced to the 
{\sc Min-Cost $k$-Flow} problem, which has a polynomial time algorithm.  

Similarly, one of the most fundamental problems in the activation setting is the {\sc Activation $st$-Path} problem. 
For the min-power version, a linear time reduction  to the ordinary {\sc Shortest $st$-Path} problem is given by Althaus et al. \cite{ACMP}. 
Lando and Nutov \cite{LN} suggested a more general (but less efficient) ''levels reduction'' 
that converts several power problems into problems with node costs; this method extends also to the activation setting, see \cite{N-sur}. 
A fundamental generalization is activating a set of $k$ internally disjoint or edge disjoint $st$-paths. 
Formally, the internally disjoint $st$-paths version is as follows. 

\begin{center} \fbox{\begin{minipage}{0.98\textwidth} \noindent
\underline{{\sc Activation $k$ Disjoint $st$-Paths} ({\sc Activation $k$-DP})} \\
{\em Input:} \ A multi-graph $G=(V,E)$ with activation costs $\{c_e^u,c_e^v\}$ for each $uv$-edge $e \in E$, 
\hphantom{Input: } $s,t \in V$, and an integer $k$. \\
{\em Output:}  An edge set $F \subs E$  of $k$ internally disjoint $st$-paths of minimum activation cost.
\end{minipage}}\end{center}


{\sc Activation $k$-DP} admits an easy approximation ratio $2$, c.f. \cite[Corollary~15.4]{N-sur}
and is polynomially solvable on bounded treewidth graphs \cite{AE-tw}.
{\sc Node-Weighted Activation $k$-DP} admits a polynomial time algorithm, 
by a reduction to the ordinary {\sc Min-Cost $k$-DP}.  
However, the complexity status of {\sc Min-Power $k$-DP} is open even for unit power costs -- 
it is not known whether the problem is in P or is NPC; this is so even for $k=2$. 

In the augmentation version of the problem {\sc Activation $k$-DP Augmentation} , 
we are also given a subgraph $G_0=(V,E_0)$ of $G$ of activation cost zero that 
already contains $k-1$ disjoint $st$-paths, and seek an augmenting edge set 
$F \subs E \sem E_0$ such that $G_0 \cup F$ contains $k$ disjoint $st$-paths.
The following lemma was implicitly proved in \cite{AE}. 

\begin{lemma} [\cite{AE}] \label{l:folk}
If {\sc Activation $2$-DP Augmentation} admits a polynomial time algorithm then {\sc Activation $2$-DP}  admits approximation ratio $1.5$.
\end{lemma}

The justification of Lemma~\ref{l:folk} is as follows. 
We may assume that $c_e^s=0$ and $c_e^t=0$ for every edge $e$ incident to $s$ or to $t$, respectively. 
For this, we ``guess'' the values $l_s=\ell_{F^*}(s)$ and $l_t=\ell_{F^*}(t)$ of some optimal solution $F^*$ at $s$ and $t$, respectively;
there are at most  $\deg_G(s) \cdot \deg_G(t)$ choices, so we can try all choices and return the best outcome.
Then for every edge $e=sv \in E$, remove  $e$ if $c_e^s> l_s$ and set $c^e_s=0$ otherwise, 
and apply a similar operation on edges incident to $t$. 
One can see that the new instance is equivalent to the original one.
Since the activation cost incurred at $s$ and $t$ is now zero, 
the cheaper among the two disjoint $st$ paths of $F^*$ has activation cost at most half $\f{1}{2}{\sf opt}$, 
where ${\sf opt}$ is the optimal solution value (to the modified problem). 
Thus if we compute an optimal $st$-path $P$ and and optimal augmenting edge set for $P$,
the overall activation cost will be $\f{3}{2}{\sf opt}$. 

When the paths are required to be only edge disjoint we get the {\sc Activation $k$-EDP} problem. 
This problem admits an easy ratio $2k$. 
Lando \& Nutov \cite{LN} improved the approximation ratio to $k$ 
by showing that {\sc Min-Power $k$-EDP Augmentation} (the augmentation version of {\sc Min-Power $k$-EDP})
admits a polynomial time algorithm.
This algorithm extends to the activation case, see \cite{N-sur}.  
For simple graphs, {\sc Min-Power $k$-EDP} admits ratio $O(\sqrt{k})$ \cite{N-kedp}.
On the other hand \cite{N-snw} shows that ratio $\rho$ for {\sc Min-Power} or {\sc Node-Weighted $k$-EDP} 
with unit costs/weights implies ratio $1/2\rho^2$ for the {\sc Densest $\ell$-Subgraph} problem,
that currently has best known ratio $O(n^{-(1/4+\epsilon)})$ \cite{BCCF} 
and approximation threshold $\Omega\left(n^{-1/poly (\log \log n)}\right)$ \cite{M}.

Based on an idea of Srinivas \& Modiano \cite{SM}, Alqahtani \& Erlebach \cite{AE} 
showed that {\sc Activation $2$-EDP} is not harder to approximate than {\sc Activation $2$-DP}.

\begin{lemma} [Alqahtani \& Erlebach \cite{AE}] \label{l:SM}
If {\sc Activation $2$-DP} admits approximation $\rho$ then so is {\sc Activation $2$-EDP}. 
\end{lemma}

Alqahtani \& Erlebach \cite{AE} claimed that 
{\atdpa} admits a polynomial time algorithm,
and thus (by Lemmas \ref{l:folk} and \ref{l:SM}) 
both {\sc Activation $2$-DP} and {\sc Activation $2$-EDP} admit approximation ratio $1.5$. 
In the next section we will give an example that the approximation 
ratio of the \cite{AE} algorithm for {\atdpa} is not better than $2$. 
Then we will give a different polynomial 
algorithm for {\atdpa} that is based on dynamic programming. 
Thus combining with Lemmas \ref{l:folk} and \ref{l:SM} we have the following. 

\begin{theorem} \label{t1}
{\atdpa} admits a polynomial time algorithm. Thus both {\sc Activation $2$-DP} and {\sc Activation $2$-EDP} admit approximation ratio $1.5$. 
\end{theorem}

\section{A bad example for the Alqahtani-Erlebach Algorithm} \label{s:AE}

To illustrate the idea of the \cite{AE} algorithm, let us first describe a known algorithm 
for a particular case of the {\sc Min-Cost $2$-EDP Augmentation} problem, 
where we seek to augment a Hamiltonian $st$-path $P$ of cost $0$ 
by an min-cost edge set $F$ such that $P \cup F$ contains $2$ edge disjoint $st$-paths. 
The algorithm reduces this problem to the ordinary {\sc Min-Cost $st$-Path} problem as follows, 
see Fig.~\ref{f:split}(a,b). 

\medskip 

\begin{algorithm}[H] 
\caption{{\sc Hamiltonian Min-Cost $2$-EDP Augmentation}$((V,E),c,P,\{s,t\})$} \label{alg:cost}
Construct an edge-weighted digraph $D_P$ by directing $P$ ``backward'' from $t$ to $s$, 
and directing every edge not in $P$ ''forward'' -- from predecessor to successor in $P$.   \\
Compute a shortest $st$-path $P'$ in $D_P$.                                                                          \\
Return the subset $F$ of $E$ that corresponds to the edges of $P' \sem P$.
\end{algorithm}

\medskip 

A slight modification of this algorithm works for {\sc Activation $2$-EDP Augmentation}. 
For $v \in V$ let $L_v=\{c_{vu}^v:vu \in E\}$ be the set of possible {\bf levels} at $v$. 
Apply the reduction in Algorithm~\ref{alg:cost}, and then apply a step which we call {\bf Levels Splitting}:
for every  pair $(v,l)$ where $v \in V$ and $l \in L_v$ we add a node $v_l$ of weight $l$, 
and put an edge from $u_{l_i}$ to $v_{l_j}$ if there is an edge $e=uv$ in $D_P$ with $c_e^u \leq l_i$ and $c_e^v \leq l_j$.  
The reduction here is to the {\sc Node-Weighted $st$-Path} problem.
The later problem can be easily reduced to the ordinary {\sc Min-Cost $st$-Path} problem by a step 
which we call  {\bf In-Out Splitting}: Replace each node $v \in V \sem \{s,t\}$ 
by two nodes $v^{\sf in},v^{\sf out}$ connected by the edge $v^{\sf in}v^{\sf out}$,
and redirect every edge that enters $v$ to enter $v^{\sf in}$ and every edge that leaves $v$ to leave $v^{\sf out}$,
where we assume that $s^{\sf in}=s^{\sf out}=s$ and $t^{\sf in}=t^{\sf out}=t$.
In this reduction the cost/weight of each edge $v^{\sf in}v^{\sf out}$ is the weight $w_v$ of $v$, see Fig.~\ref{f:split}(a,b,c).
This is a particular case of the ``Levels Reduction'' of \cite{LN}. 

\begin{figure}
\centering 
\includegraphics{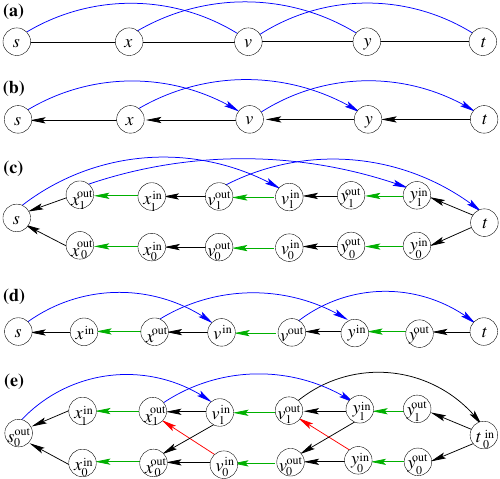}
\caption{Augmenting a Hamiltonian $st$-path to two edge/internally disjoint $st$-paths. 
Black edges have cost $0$, blue and red edges have cost $1$.
(a) Problem instance.
(b) Reducing {\sc Min-Cost $2$-EDP Augmentation} to {\sc Min-Cost $st$-Path}.
(c) Levels splitting, assuming that the activation costs of the blue edges in (a) are $1$ at $x,v,y$ and $0$ at $s,t$. 
(d) Reducing {\sc Min-Cost $2$-DP Augmentation} to {\sc Min-Cost $st$-Path}.
(e) The reduction of \cite{AE}.}
\label{f:split}
\end{figure}

One can also solve the 
version when we have ordinary edge costs and 
require that $P \cup F$ contains $2$ internally disjoint $st$-paths. 
For that, apply a standard reduction that converts edge connectivity problems into node connectivity ones, as follows
\begin{enumerate}[(i)]
\item
After step~1 of Algorithm~\ref{alg:cost}, add the In-Out Splitting step, 
where here the the cost of each edge $v^{\sf in}v^{\sf out}$ is $0$.
\item
Replace every edge $u^{\sf out}v^{\sf in} \notin P$ by the edge $u^{\sf in}v^{\sf out}$. 
\end{enumerate}
See  Fig.~\ref{f:split}(d), where after applying this reduction we switched between the names of $v^{\sf in}v^{\sf out}$, 
to be consistent with the \cite{AE} algorithm. 

The algorithm of \cite{AE} attempts to combine the later reduction with the Levels Reduction in a sophisticated way. 
In the case when 
$G$ has a zero cost Hamiltonian $st$-path $P$, 
$st \notin E$, $L=\{0,1\}$, $L_s=L_t=\{0\}$,
and $c_{uv}^u=1$ for all $e=uv \in E\sem E(P)$ and $u \in V \sem \{s,t\}$,
the \cite{AE} algorithm reduces to the following, see Fig.~\ref{f:split}(a,e).
\begin{enumerate}
\item
Construct an edge-weighted directed graph $D_P$ with nodes $s=s_0^{\sf out},t=t_0^{\sf in}$ and $4$ nodes 
$\{v_0^{\sf in},v_0^{\sf out}, v_1^{\sf in},v_1^{\sf out}\}$ for every $v \in V \sem \{s,t\}$. 
The edge of $D_P$ and their weights are:
\begin{enumerate}[(i)]
\item
For $v \in V \sem \{s,t\}$: \hspace{2.8cm}
$w(v_a^{\sf out} v_a^{\sf in})=0$ \hspace{0.15cm} $a \in \{0,1\}$.
\item
For $uv \in P$: \hspace{3.75cm}
$w(v_b^{\sf in} u_a^{\sf out})=a$ \  \ $a,b \in \{0,1\}$. \\
For $uv \notin P$: \hspace{3.75cm}
$w(u_a^{\sf out} v_b^{\sf in})=b$ \hspace{0.15cm} $a,b \in \{0,1\}$, $u_a v_b \in E$.
\end{enumerate}
\item
Compute a cheapest $st$-path $P'$ in $D_P$ and return the subset of $E$ that corresponds to $P'$. 
\end{enumerate}

\begin{figure}
\centering 
\includegraphics{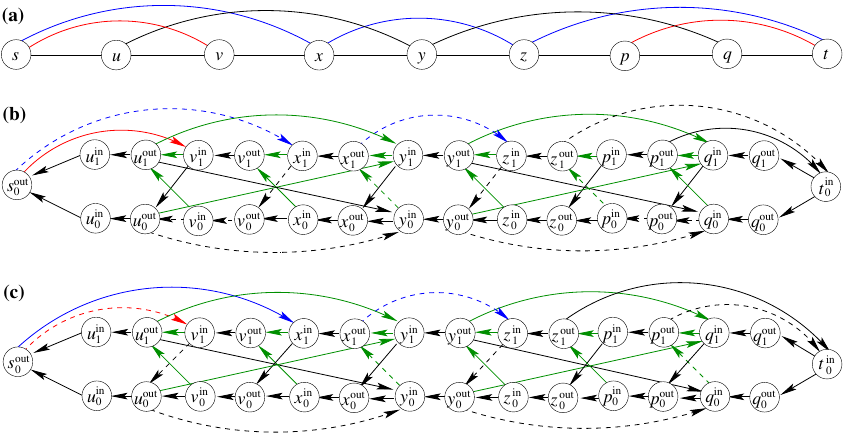}
\caption{
Illustration to the \cite{AE} Algorithm. Black edges have weight/thresholds $0$. 
(a) The input graph; colored edges have thresholds $0$ at $s,t$ and $1$ otherwise.
(b) The edge weighted directed graph $D_P,w$ constructed in the AE reduction and
the path (shown by dashed lines) in $D_P$ of weight $4$ that corresponds to the optimal solution $\{sx,xz,zt,uy,yq\}$. 
(c) The path (shown by dashed lines) in $D_P$ of weight $4$ that corresponds to the solution $\{sv,uy,xz,yq,pt\}$.
}
\label{f:01}
\end{figure}

Here for $e=uv \in E$ we write $u_a v_b \in E$ meaning that $c_e^a \leq a$ and $c_e^b \leq b$, namely,
that $uv$ can be activated by assigning $a$ units to $u$ and $b$ units to $v$.
In the example in Fig.~\ref{f:split}(d), the weight of $P'$ is $2$ while the optimal solution value is $3$. 
Still, in this example the \cite{AE} algorithm computes an optimal solution. 
We give a more complicated example, which shows that the approximation ratio of this algorithm is no better than $2$. 
Consider the graph in Fig.~\ref{f:01}(a), with the initial $st$-path 
$$
s-u-v-x-y-z-p-q-t \ .
$$
The optimal solution $\{sx,xz,zt,,uy,yq\}$ (the blue edges and the $uy,yq$ edges) has value $2$ (level assignment $l_x=l_z=1$ and $0$ otherwise);
the $s_0^{\sf out} t_0^{\sf in}$-path in $D_P$ of weight $4$ that corresponds to this solution is (see Fig.~\ref{f:01}(b)):
$$
s_0^{\sf out} 
{\color{blue} \rightarrow} \ x_1^{\sf in} \rightarrow v_0^{\sf out} \rightarrow v_0^{\sf in} \rightarrow u_0^{\sf out} \rightarrow y_0^{\sf in} 
{\color{olive} \rightarrow} \ x_1^{\sf out} 
{\color{blue} \rightarrow} \ z_1^{\sf in} \rightarrow y_0^{\sf out} \rightarrow \ q_0^{\sf in} \rightarrow p_0^{\sf out} \rightarrow p_0^{\sf in} 
{\color{olive} \rightarrow} \ z_1^{\sf out} \rightarrow \ t_0^{\sf in} \ .
$$
The solution $\{sv,uy,xz,yq,pt\}$ (the red edges and the $uy,yq,xz$ edges) has value $4$ (level assignment $l_v=l_x=l_z=l_p=1$ and $0$ otherwise);
the $s_0^{\sf out} t_0^{\sf in}$-path in $D_P$ of weight $4$ that corresponds to this solution is (see Fig.~\ref{f:01}(c)):
$$
s_0^{\sf out} 
{\color{red}   \rightarrow} \ v_1^{in} \rightarrow u_0^{\sf out} \rightarrow y_0^{\sf in} 
{\color{olive} \rightarrow} \ x_1^{\sf out} 
{\color{blue} \rightarrow} \ z_1^{\sf in} \rightarrow y_0^{\sf out} \rightarrow q_0^{\sf in} 
{\color{olive} \rightarrow} \ p_1^{\sf out} \rightarrow t_0^{\sf in} \ .
$$
So in $D_P$, both paths have the same weight $4$, 
but one path gives a solution of value $2$ while the other of value $4$. 

\section{Proof of Theorem~\ref{t1}} \label{s:1}

In this section we will prove Theorem~\ref{t1} -- that {\atdpa} admits a polynomial time algorithm

Recall that for $F \subs E$ and $v \in V$ we denote by $\displaystyle \ell_F(v)=\max_{e \in \delta_F(v)}c_e^v$ 
the activation cost incurred by $F$ at $v$,
and that for $S \subs V$ the activation cost incurred by $F$ at nodes in $S$ is
\[
\displaystyle \ell_F(S)=\sum_{v \in S} \ell_F(v) =\sum_{v \in S}\max_{e \in \delta_F(v)}c_e^v \ .
\]

For the proof of Theorem~\ref{t1} it would be convenient to consider a more general problem 
where each edge $e=uv \in E$ has three costs $c_e^u,c(e),c_e^v$, where 
$c_e^u,c_e^v$ are the activation costs of $e$ and $c(e)$ is the ordinary ``middle'' cost of $e$. 
We now describe a method to convert an {\atdpa} instance into an equivalent instance 
in which $P$ is a Hamiltonian path but every edge has three costs as above. 
We call this problem {\sc 3-Cost Hamiltonian} {\atdpa}.

Let $\II=(G=(V,E),c,s,t,P)$ be an {\atdpa} instance. 
Let us say that a $uv$-path $Q$ in $G$ is an {\bf attachment path} if $u,v \in P$ but $Q$ has no internal node in $P$.
Note that any inclusion minimal edge set that contains $2$ internally disjoint $st$-paths is a cycle. 
This implies that if $F$ is an inclusion minimal solution to {\atdpa} then $\deg_F(v) \in \{0,2\}$ for every node $v \in V \sem V(P)$,
hence $F$ partitions into attachment paths. 
This enables us to apply a prepossessing similar to metric completion, and to construct  
an equivalent {\sc 3-Cost Hamiltonian} {\atdpa} instance $\hat \II=(\hat G=(\hat V,\hat E),\hat c, s,t,P)$.
For this, for every $u,v \in P$ and $(l_u,l_v) \in L_u \times L_v$ do the following.
\begin{enumerate}
\item
Among all attachment $uv$-paths that have activation costs $l_u$ at $u$ and $l_v$ at $v$ (if any), 
compute the cheapest one $Q(l_u,l_v)$.
\item
If $Q(l_u,l_v)$ exists, add a new edge $e=uv$ with activation costs $\hat c_e^u=l_u, \hat c_e^v=l_v$, 
and ordinary cost $\hat c_e= \ell_{Q(l_u,l_v)}(V)-(l_u+l_v)$ being the activation cost of $Q(l_u,l_v)$ on internal nodes of $Q(l_u,l_v)$.
\end{enumerate}
After that, remove all nodes in $V \sem V(P)$. Now $P$ is a Hamiltonian path, and we get a 
{\sc 3-Cost Hamiltonian} {\atdpa} instance $\hat \II=(\hat G,\hat c,s,t,P)$.
It is easy to see that the instance $\hat \II$ can be constructed in polynomial time. 
Note that the instance $\II'$ may have many parallel edges, but this is allowed, also in the original instance $\II$. 

Now consider some feasible solution $F$ to $\II$. 
Replacing every attachment paths contained in $F$ by a single edge as in step 2 above
gives a feasible solution $\hat F$ to $\hat \II$ of value at most that of $F$.  
Conversely, if $\hat F$ is a feasible $\hat \II$ solution, then replacing every edge in $\hat F$ 
by an appropriate path gives a feasible solution $F$ to $\II$ of value at most that of $\hat F$.  
Consequently, the new instance is equivalent to the original instance in the sense that 
every feasible solution to one of the instances can be converted to a feasible solution to the other instance of no greater value. 
We summarize this as follows.

\begin{corollary}
If {\sc 3-Cost Hamiltonian} {\atdpa} admits a polynomial time algorithm 
then so is {\atdpa}.
\end{corollary}

So from now and on our problem is {\sc 3-Cost Hamiltonian} {\atdpa}. 
Let us denote by $\tau(F)$ the sum of the ordinary and the activation cost of $F \subs E$, namely
\[
\tau(F)=\ell_F(V)+c(F)=\sum_{v \in V}\max_{e \in \delta_F(v)}c_e^v+\sum_{e \in F} c(e) \ .
\]
Let ${\sf opt}$ denote an optimal solution value for an instance of this problem. 
We will assume that $V=\{0,1, \ldots,n\}$ and that $P=0-1 - \cdots - n$ is a (Hamiltonian) $(0,n)$-path,
and view each edge $ij \notin P$ as a directed edge $(i,j)$ where $i<j$. 
Our goal is to find and edge set $F \subset E$ such that $P \cup F$ contains $2$ internally disjoint $(0,n)$-paths 
and such that $\tau(F)$ is minimal.

\begin{definition}
For $0 \leq i < j< n$ let $\FF_{i,j}$ denote the family of all edge sets $F \subs E$ that satisfy the following 
two conditions. 
\begin{enumerate}[(i)]
\item 
$F$ is an inclusion minimal edge set such that $P \cup F$ contains $2$ internally disjoint $(j-1,n)$-paths.
\item
No edge in $F$ has an end strictly preceding $i$, namely, if $(x,y) \in F$ then $x \geq i$.
\end{enumerate}
\end{definition}

We will need the following (essentially known) ``recursive'' property of the sets in $\FF_{i,j}$.

\begin{lemma} \label{l:F}
$F \in \FF_{i,j}$ if and only if there exists $i \leq x <j$ such that exactly one of the following holds, see Fig.~\ref{f:st2}(a). 
\begin{enumerate}[(i)]
\item
$F=\{(x,n)\}$. 
\item
$F=\{(x,y)\} \cup F'$ for some $(x,y) \in E$ with $i \leq x<j<y < n$ and $F' \in \FF_{j,y}$.
\end{enumerate}
\end{lemma}
\begin{proof}
It is easy to see that if $|F|=1$ then (i) must hold. Assume that $|F| \geq 2$.  
There is $(x,y) \in F$ with $x<j<y$ as otherwise $(P \cup F) \sem \{j\}$ has no $(j-1,n)$-path.
Let $F'= F \sem \{(x,y)\}$.
Then $P \cup F'$ contains $2$ internally disjoint $(y-1,n)$-paths, 
as otherwise $(P \cup F) \sem \{y\}$ has no $(j-1,n)$-path.
Let $x'$ be the lowest end of an edge in $F'$, let $(x',y') \in F'$, and let $F''=F \sem \{(x',y')\}$. 
If $x' < j$ then $F' \in \FF_{i,j}$ (if $y' \geq y$) or $F'' \in \FF_{i,j}$ (if $y' \leq y$), contradicting the minimality of $\FF$.
If $x'\geq y$ then  $(P \cup F) \sem \{y\}$ has no $(j-1,n)$-path.
This implies that $F' \in \FF_{i,j}$, hence (ii) holds. 
\end{proof}

\begin{figure}
\centering 
\includegraphics{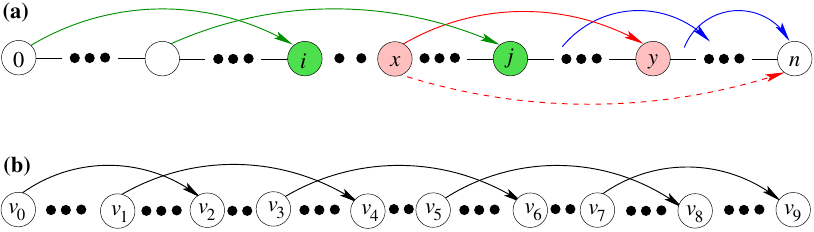}
\caption{Illustration to Lemma~\ref{l:F}. 
Three dots between nodes indicates that these nodes are distinct, while if there are only two dotes then the nodes may coincide.}
\label{f:st2}
\end{figure}

We note that Lemma~\ref{l:F} has the following (essentially known) consequence. 
Consider an inclusion minimal feasible solution $F$ to our problem. 
Then the edges in $F$ have an order
$(v_0,v_2),(v_1,v_4),(v_3,v_6), \ldots, (v_{2q-1},v_{2q+1})$ such that (see Fig.~\ref{f:st2}(b)) 
\[
0=v_0<v_1<v_2 \leq v_3 < v_4 \leq v_5 < \cdots \leq v_{2q-1} < v_{2q}<v_{2q+1}=n \ .
\]
Note that in this node sequence some nodes may be identical (e.g., we may have $v_2=v_3$), 
while the others are required to be distinct (e.g., $v_0<v_1<v_2$).

\medskip

For $(l_i,l_j) \in L_i \times L_j$ 
let $E(l_i,l_j)=\{e \in E: \ell_e(l_i) \leq l_i, \ell_e(l_j) \leq l_j\}$.
Let $0 \leq i <j< n$.
For $F \subs E$ the {\bf $(l_i,l_j)$-forced cost of $F$} is defined by 
\[
\al_F(l_i,l_j)= \left \{ \begin{array}{ll}
c(F)+\ell_F(V \sem \{i,j\})+l_i+l_j      & \ \mbox{if} \ \ F \subs E(l_i,l_j) \\
 \infty                                                   & \mbox{ otherwise}
\end{array} \right .
\]
Namely, assuming $F \subs E(l_i,l_j)$ 
we pay $\ell_F(v)$ at every node $v \in V \sem \{i,j\}$, 
and in addition we ``forcefully'' pay $l_i$ at $i$ and $l_j$ at $j$.
Note that 
\[
\al_F(l_i,l_j)=c(F)+\ell_F(V \sem \{i,j\})+l_i+l_j = \tau(F)+(l_i-\ell_F(i))+(l_j-\ell_F(j))  \ .
\]
This implies the following. 

\begin{corollary} \label{c:al}
$\al_F(l_i,l_j) \geq \tau(F)$ and an equality holds if and only if $\ell_F(i)=l_i$ and $\ell_F(j)=l_j$.
\end{corollary}

Let $f(l_i,l_j)$ denote the minimal $(l_i,l_j)$-forced cost of an edge set $F \in \FF_{i,j}$, namely
\begin{equation} \label{e:f}
f(l_i,l_j) = \min_{F \in \FF_{i,j}} \al_F(l_i,l_j) \ .
\end{equation} 

The number of possible values of $f(l_i,l_j)$ is $O(|V|^2|L|^2)$ - there are $O(|V|^2)$ choices of $i,j$ and at most $|L|^2$ choices of $l_i,l_j$.
We will show how to compute all these values in polynomial time using dynamic programming. 
Specifically, we will get a recursive formula that enables to compute each value either directly, or using previously computed values. 

The next lemma shows how the function $f(l_i,l_j)$ is related to our problem. 

\begin{lemma}
${\sf opt}=\displaystyle{\min_{l_0,l_1 \in L}f(l_0,l_1)}$.
\end{lemma}
\begin{proof}
From the definition of $\FF_{i,j}$ it follows that that  $\FF_{0,1}$ is the set of all inclusion minimal feasible solutions.
Let $F^*$ be an inclusion minimal optimal solution and let $\hat F$ be the minimizer of (\ref{e:f}).
Then for any $l_0,l_1$ we have
\[
f(l_0,l_1)=\al_{\hat F}(l_0,l_1) \geq \tau(\hat F) \geq \tau(F^*) = {\sf opt} \ .
\]
Consequently, $\displaystyle{\min_{l_0,l_1 \in L}f(l_0,l_1)} \geq {\sf opt}$.
On the other hand for $l_0=\ell_{F^*}(0)$ and  $l_1=\ell_{F^*}(1)$ we have 
\[
f(l_0,l_1)=\al_{F^*}(l_0,l_1)=\tau(F^*)={\sf opt} \ ,
\] 
by Corollary~\ref{c:al}. This implies $\displaystyle{\min_{l_0,l_1 \in L}f(l_0,l_1)} \leq {\sf opt}$, concluding the proof.
\end{proof}

When $F=\{e\}$ is a single edge we will use the abbreviated notation $\al_e(l_i,l_j)=\al_{\{e\}}(l_i,l_j)$. 
For  $(l_i,l_y) \in L_i \times L_y$ and an edge $e=(x,y) \in E(l_i,l_y)$ with $i \leq x <y$ let $\be_e(l_i,l_y)$ be defined by
\[
\be_e(l_i,l_y)= \left \{ \begin{array}{ll}
0                     & \mbox{if} \ x=i \\
 \ell_e(x)        & \mbox{if} \ x>i
\end{array} \right .
\]
We now define two functions that reflect the two different scenarios in Lemma~\ref{l:F}. 
\begin{eqnarray*}
g(l_i,l_j) & = & \min_e\{\al_e(l_i,l_j) : e=(x,n) \in E, i \leq x <j\} \\
h(l_i,l_j) & = & \min_{l_y,e}\{c(e)+l_i+\be_e(l_i,l_y)+f(l_j,l_y):  e=(x,y) \in E(l_i,l_y), i \leq x <j<y\}
\end{eqnarray*}

It is not hard to see that the function $g(l_i,l_j)$ is the minimal $(l_i,l_j)$-forced cost of a single edge set $e=(x,n)$ such that 
$\{e\} \in \FF_{i,j}$.
Thus if there exist a minimizer of (\ref{e:f}) that is a single edge, then $f(l_i,l_j)=g(l_i,l_j)$. 

We will show that the function $h(l_i,l_j)$ is the minimal $(l_i,l_j)$-forced cost of a non-singleton set $F \in \FF_{i,j}$;
note that by lemma~\ref{l:F} any such $F$ is a union of some single edge $(x,y) \in F$ with $i \leq x<j<y < n$ and $F' \in \FF_{j,y}$.
We need the following lemma.

\begin{lemma} \label{l:e}
Let $F \in \FF_{i,j}$ such that $F \subs E(l_i,l_j)$ and $|F| \geq 2$.
Let $(x,y) \in F$ be the first edge of $F$ as in Lemma~\ref{l:F}, let $F'=F \sem \{e\}$, and let $l_y=\ell_F(y)$. 
Then 
\[
\al_F(l_i,l_j) =c(e)+l_i+\be_e(l_i,l_y)+\al_{F'}(l_j,l_y)
\]
\end{lemma}
\begin{proof}
One can verify that for $l_y=\ell_F(y)$ we have $F' \subs E(l_j,l_y)$ and the following holds:
\[
 \ell_F(V \sem \{i,j\}) = \be_e(l_i,l_y)+\ell_{F'}(V \sem \{j,y\})+l_y  \ .
\]
From this and using that $c(F)=c(e)+c(F')$ we get 
\begin{eqnarray*}
\al_F(l_i,l_j) & = & c(F)+\ell_F(V \sem \{i,j\})+l_i+l_j  \\
                     & = & c(e)+l_i+\be_e(l_i,l_y)+c(F')+\ell_{F'}(V \sem \{j,y\})+l_y+l_j  \\
										 & = & c(e)+l_i+\be_e(l_i,l_y)+\al_{F'}(l_j,l_y) \ , 
\end{eqnarray*}
as required.
\end{proof}


\begin{lemma} \label{l:F*}
Among all non-singleton sets in $\FF_{i,j}$ let $F^*$ have minimal $(l_i,l_j)$-forced cost.
Then $h(l_i,l_j)=\al_{F^*}(l_i,l_j)$.
\end{lemma}
\begin{proof}
We show that $\al_{F^*}(l_i,l_j) \geq h(l_i,l_j)$. 
Let $e=(x,y) \in F^*$ be the first edge of $F^*$ as in Lemma~\ref{l:F}(ii), let $F'=F \sem \{e\}$, and let $l_y=\ell_{F^*}(y)$. 
Then $e \in E(l_i,l_y)$, hence by Lemma~\ref{l:e}
\[
\al_{F^*}(l_i,l_j) = c(e)+l_i+\be_e(l_i,l_y)+\al_{F'}(l_j,l_y) \geq h(l_i,l_j) \ .
\]
The inequality is since in the definition of $h(l_i,l_j)$ we minimize over $e=(x,y)$ and $l_y$. 

We show that $\al_{F^*}(l_i,l_j) \leq h(l_i,l_j)$. 
Let $e=(x,y) \in E$ and $l_y$ be the parameters for which 
the minimum in the definition of $h(l_i,l_j)$ is attained, 
and let $F' \in \FF_{j,y}$ such that $f(l_j,l_y) = \al_{F'}(l_j,l_y)$.
Let $F=F' \cup \{e\}$ and note that $F \in \FF_{i,j}$ (by Lemma~\ref{l:F}) and that $l_y=\ell_F(y)$ (by the definition of $h(l_i,l_j)$). 
Consequently, 
\[
 \al_{F^*}(l_i,l_j) \leq  \al_F(l_i,l_j) = c(e)+l_i+\be_e(l_i,l_y)+f(l_j,l_y) = h(l_i,l_j) 
\]
The inequality is since $F^*$ has minimal $(l_i,l_j)$-forced cost. 

We showed that $\al_{F^*}(l_i,l_j) \geq h(l_i,l_j)$ and that $\al_{F^*}(l_i,l_j) \leq h(l_i,l_j)$,
hence the proof is complete. 
\end{proof}

Let $F^*$ be the minimizer of (\ref{e:f}). From Lemma~\ref{l:F*} we have:
\begin{itemize}
\item
If $|F^*|=1$ then $f(l_i,l_j)=g(l_i,l_j)$.
\item
If $|F^*|\geq 2$ then $f(l_i,l_j)=h(l_i,l_j)$.
\end{itemize}
Therefore
\begin{equation} \label{e:rec}
f(l_i,l_j)=\min\{g(l_i,l_j),h(l_i,l_j)\}
\end{equation}
Note that the quantities $g(l_i,l_j)$ can be computed directly in polynomial time. 
The recurrence in (\ref{e:rec}) enables to compute values of $f(l_i,l_j)$, 
for all $0 \leq i< j \leq n-1$ and $(l_i,l_j) \in L_i \times L_j$, in polynomial time.  
The number of such values is $O(n^2 |L|^2)$, concluding the proof of Theorem~\ref{t1}. 

\medskip

Let us illustrate the recursion by showing how the values of $f(l_i,l_j)$ are computed for $j=n-1,n-2$. 
Recall that the values of the function $g$ are computed directly, without recursion and that
\[
h(l_i,l_j) =  \min_{l_y,e}\{c(e)+l_i+\be_e(l_i,l_y)+f(l_j,l_y):  e=(x,y) \in E(l_i,l_y), i \leq x <j<y\}
\]

For $j=n-1$ we have $h(l_i,l_j)=\infty$, and thus:
\[
f(l_i,l_{n-1}) =g(l_i,l_{n-1})
\]

\medskip

For $j=n-2$, the only possible value of $y$ is $y= n-1$.
For every $i \leq x <y=n-1$ and $e=xy \in E$ we compute directly (without recursion) the values $\be_e(l_i,l_y)$.
Then
\begin{eqnarray*} \displaystyle 
h(l_i,l_{n-2}) & = & \min_{l_y,e}\{c(e)+l_i+\be_e(l_i,l_y)+f(l_j,l_y): e=(x,y) \in E(l_i,l_y), i \leq x <j<y\} \\
                      & = & \min_{l_{n-1},e} \left\{ c(e)+l_i+\be_e(l_i,l_{n-1})+f(l_{n-2},l_{n-1}): \begin{array}{ll}
                                                                                                                                               e=(x,n-1) \in E(l_i,l_{n-1}) \\
                                                                                                                                                i \leq x <j<n-1
																																																																							  \end{array} \right\}
\end{eqnarray*}
Substituting the already computed value $f(l_{n-2},l_{n-1}) =g(l_{n-2},l_{n-1})$ enables to compute the minimum
of the obtained expression, and thus also to compute $f(l_i, l_n-2)$ via (\ref{e:rec}). 

In a similar way we can compute $h(l_i, l_{n-3})$, then $f(l_i,l_{n-3})$, and so on. 


\begin{thebibliography}{10}

\bibitem{AE}
H.~M. Alqahtani and T.~Erlebach.
\newblock Approximation algorithms for disjoint $st$-paths with minimum
  activation cost.
\newblock In {\em CIAC}, pages 1--12, 2013.

\bibitem{AE-tw}
H.~M. Alqahtani and T.~Erlebach.
\newblock Minimum activation cost node-disjoint paths in graphs with bounded
  treewidth.
\newblock In {\em SOFSEM}, pages 65--76, 2014.

\bibitem{ACMP}
E.~Althaus, G.~Calinescu, I.~Mandoiu, S.~Prasad, N.~Tchervenski, and
  A.~Zelikovsky.
\newblock Power efficient range assignment for symmetric connectivity in static
  ad-hoc wireless networks.
\newblock {\em Wireless Networks}, 12(3):287--299, 2006.

\bibitem{BCCF}
A.~Bhaskara, M.~Charikar, E.~Chlamtac, U.~Feige, and A.~Vijayaraghavan.
\newblock Detecting high log-densities: an {$O(n^{1/4})$} approximation for
  densest $k$-subgraph.
\newblock In {\em STOC}, pages 201--210, 2010.

\bibitem{GH}
S.~Guha and S.~Khuller.
\newblock Approximation algorithms for connected dominating sets.
\newblock {\em Algorithmica}, 20:374--387, 1998.

\bibitem{HKMN}
M.~Hajiaghayi, G.~Kortsarz, V.~Mirrokni, and Z.~Nutov.
\newblock Power optimization for connectivity problems.
\newblock {\em Math. Program.}, 110(1):195--208, 2007.

\bibitem{KKKP}
L.~M. Kirousis, E.~Kranakis, D.~Krizanc, and A.~Pelc.
\newblock Power consumption in packet radio networks.
\newblock {\em Theoretical Computer Science}, 243(1-2):289--305, 2000.

\bibitem{KR}
P.~Klein and R.~Ravi.
\newblock A nearly best-possible approximation algorithm for node-weighted
  {Steiner} trees.
\newblock {\em J. Algorithms}, 19(1):104--115, 1995.

\bibitem{KNS}
G.~Kortsarz, Z.~Nutov, and E.~Shalom.
\newblock Approximating activation edge-cover and facility location problems.
\newblock {\em Theoretical Computer Science}, 930:218--228, 2022.

\bibitem{LN}
Y.~Lando and Z.~Nutov.
\newblock On minimum power connectivity problems.
\newblock {\em J. Discrete Algorithms}, 8(2):164--173, 2010.

\bibitem{M}
P.~Manurangsi.
\newblock Almost-polynomial ratio {ETH}-hardness of approximating densest
  $k$-subgraph.
\newblock In {\em STOC}, pages 954--961, 2017.

\bibitem{N-powcov}
Z.~Nutov.
\newblock Approximating minimum power covers of intersecting families and
  directed edge-connectivity problems.
\newblock {\em Theoretical Computer Science}, 411(26-28):2502--2512, 2010.

\bibitem{N-snw}
Z.~Nutov.
\newblock Approximating steiner networks with node-weights.
\newblock {\em SIAM J. Comput.}, 39(7):3001--3022, 2010.

\bibitem{N-sur}
Z.~Nutov.
\newblock Activation network design problems.
\newblock In T.~F. Gonzalez, editor, {\em Handbook on Approximation Algorithms
  and Metaheuristics, Second Edition}, volume~2, chapter~15. Chapman \&
  Hall/CRC, 2018.

\bibitem{N-kedp}
Z.~Nutov.
\newblock An $o(\sqrt{k})$-approximation algorithm for minimum power $k$ edge
  disjoint $st$-paths.
\newblock {\em CoRR}, abs/2208.09373, 2022.
\newblock To appear in CIE 2023.
\newblock URL: \url{https://doi.org/10.48550/arXiv.2208.09373}.

\bibitem{P}
D.~Panigrahi.
\newblock Survivable network design problems in wireless networks.
\newblock In {\em SODA}, pages 1014--1027, 2011.

\bibitem{RM}
V.~Rodoplu and T.~H. Meng.
\newblock Minimum energy mobile wireless networks.
\newblock In {\em IEEE International Conference on Communications (ICC)}, pages
  1633--1639, 1998.

\bibitem{S}
A.~Segev.
\newblock The node-weighted {Steiner} tree problem.
\newblock {\em Networks}, 17:1--17, 1987.

\bibitem{SRS}
S.~Singh, C.~S. Raghavendra, and J.~Stepanek.
\newblock Power-aware broadcasting in mobile ad hoc networks.
\newblock In {\em Proceedings of IEEE PIMRC}, 1999.

\bibitem{SM}
A.~Srinivas and E.~H. Modiano.
\newblock Finding minimum energy disjoint paths in wireless ad-hoc networks.
\newblock {\em Wireless Networks}, 11(4):401--417, 2005.

\bibitem{WNE}
J~E. Wieselthier, G.~D. Nguyen, and A.~Ephremides.
\newblock On the construction of energy-efficient broadcast and multicast trees
  in wireless networks.
\newblock In {\em Proc. IEEE INFOCOM}, pages 585--594, 2000.

\end{thebibliography}

\end{document}